\documentclass[11pt]{article}      
\usepackage{amsmath,amssymb,amscd,amsthm, graphicx, verbatim, setspace} 
\usepackage{authblk}

\theoremstyle{plain}
\newtheorem{thm}{Theorem}[section]

\newtheorem{prop}[thm]{Proposition}

\newtheorem{conj}[thm]{Conjecture}
\newtheorem{ques}[thm]{Question}
\usepackage{mathtools}

\newcommand{\Ex}{\operatorname{Ex}}
\newcommand{\F}{\operatorname{F}}

\newcommand{\fw}{\operatorname{fw}}
\newcommand{\fl}{\operatorname{fl}}

\title{An algorithm for bounding extremal functions of forbidden sequences}
\date{}
\author{{\bf{Jesse Geneson}}\\
\small Iowa State University, Ames, IA 50011, USA\\
{\small\em geneson@gmail.com}}
\begin{document}
\maketitle

\begin{abstract}
Generalized Davenport-Schinzel sequences are sequences that avoid a forbidden subsequence and have a sparsity requirement on their letters. Upper bounds on the lengths of generalized Davenport-Schinzel sequences have been applied to a number of problems in discrete geometry and extremal combinatorics. Sharp bounds on the maximum lengths of generalized Davenport-Schinzel sequences are known for some families of forbidden subsequences, but in general there are only rough bounds on the maximum lengths of most generalized Davenport-Schinzel sequences. One method that was developed for finding upper bounds on the lengths of generalized Davenport-Schinzel sequences uses a family of sequences called formations. 

An $(r, s)$-formation is a concatenation of $s$ permutations of $r$ distinct letters. The formation width function $\fw(u)$ is defined as the minimum $s$ for which there exists $r$ such that every $(r, s)$-formation contains $u$. The function $\fw(u)$ has been used with upper bounds on extremal functions of $(r, s)$-formations to find tight bounds on the maximum possible lengths of many families of generalized Davenport-Schinzel sequences. Algorithms have been found for computing $\fw(u)$ for sequences $u$ of length $n$, but they have worst-case run time exponential in $n$, even for sequences $u$ with only three distinct letters. 

We present an algorithm for computing $\fw(u)$ with run time $O(n^{\alpha_r})$, where $r$ is the number of distinct letters in $u$ and $\alpha_r$ is a constant that only depends on $r$. We implement the new algorithm in Python and compare its run time to the next fastest algorithm for computing formation width. We also apply the new algorithm to find sharp upper bounds on the lengths of several families of generalized Davenport-Schinzel sequences with $3$-letter forbidden patterns.
\end{abstract}

\section{Introduction}
Upper bounds on the lengths of generalized Davenport-Schinzel sequences have been used to bound the complexity of faces in arrangements of arcs with a limited number of pairwise crossings \cite{agsh}, the complexity of unions of fat triangles \cite{petties}, the complexity of lower envelopes of sets of polynomials of bounded degree \cite{DS}, and the number of edges in $k$-quasiplanar graphs with no pair of edges intersecting in more than a bounded number of points \cite{fox}. The original Davenport-Schinzel sequences of order $s$, denoted $DS(n, s)$-sequences, are sequences with $n$ distinct letters and no adjacent same letters that avoid alternations of length $s+2$. Although generalized Davenport-Schinzel sequences are not well-understood in general, sharp bounds have been found on the maximum lengths of $DS(n,s)$ sequences for both fixed $s$ and $s = \Omega(n)$ \cite{niv, pettie3, welpet, gends}. In \cite{gpt}, we developed an algorithm for obtaining upper bounds on the lengths of generalized Davenport-Schinzel sequences, giving sharp bounds for many new families of forbidden patterns. 

We say that a sequence $u$ \emph{contains} a sequence $v$ if some subsequence of $u$ is isomorphic to $v$. Otherwise $u$ \emph{avoids} $v$. A sequence is called \emph{$r$-sparse} if every contiguous subsequence of length $r$ has no repeated letters. If $u$ has $r$ distinct letters, then $\mathit{Ex}(u, n)$ is the maximum possible length of an $r$-sparse sequence that avoids $u$ with $n$ distinct letters. For example, the maximum possible length of a $DS(n, s)$-sequence is $\mathit{Ex}(u, n)$ when $u$ is an alternation of length $s+2$. Other than a few families of forbidden sequences $u$ with sharp bounds, only rough bounds are known for $\mathit{Ex}(u, n)$ in general.

An \emph{$(r, s)$-formation} is a concatenation of $s$ permutations of $r$ distinct letters. Formations were introduced in \cite{klazar1} as a tool for obtaining general upper bounds on $\mathit{Ex}(u, n)$. The extremal function $\F_{r, s}(n)$ is defined as the maximum length of an $r$-sparse sequence with $n$ distinct letters that avoids all $(r, s)$-formations. Klazar proved that if $u$ has length $s$ and $r$ distinct letters, then every $(r, s-1)$-formation contains $u$, so $\Ex(u, n) \leq \F_{r, s-1}(n)$. Nivasch improved the bound by proving that every $(r, s-r+1)$-formation contains $u$, so $\Ex(u, n) \leq \F_{r, s-r+1}(n)$ \cite{niv}. Nivasch also proved that $\F_{r, 2t-1}(n) = n2^{\frac{1}{(t-2)!}\alpha(n)^{t-2} \pm O(\alpha(n)^{t-3})}$ and $\Ex((a b)^t, n) = n2^{\frac{1}{(t-2)!}\alpha(n)^{t-2} \pm O(\alpha(n)^{t-3})}$ for all $r \geq 2$ and $t\geq 3$. The results of \cite{agshsh, klazar1, pettie, pettie3} give that $\F_{r, 4}(n) = \Theta(n \alpha(n))$, $\Ex(ababa, n) = \Theta(n \alpha(n))$, $\Ex(abababa, n) = \Theta(n \alpha(n) 2^{\alpha(n)})$, and $\Ex((ab)^t a, n) = n2^{\frac{1}{(t-2)!}\alpha(n)^{t-2} \pm O(\alpha(n)^{t-3})}$ for all $t\geq 4$. Pettie \cite{pettie3} proved that $\F_{2, 6}(n) = \Theta(n \alpha(n) 2^{\alpha(n)})$ and $\F_{2, 2t}(n) =  n2^{\frac{1}{(t-2)!}\alpha(n)^{t-2} \pm O(\alpha(n)^{t-3})}$ for all $t\geq 4$, but $\F_{r, 2t}(n) =  n2^{\frac{1}{(t-2)!}\alpha(n)^{t-2}(\log{\alpha(n)}\pm O(1))}$ for all $r, t\geq 3$.

In \cite{gpt}, we defined a function called the {\it formation width} of $u$, denoted by $\fw(u)$, to be the minimum value of $s$ such that there exists an $r$ for which every $(r, s)$-formation contains $u$. We also defined the \emph{formation length} of $u$, denoted by $\fl(u)$, to be the minimum value of $r$ for which every $(r, \fw(u))$-formation contains $u$. We found that $\Ex(u, n) = O(\F_{\fl(u),\fw(u)}(n))$ for all sequences $u$ with $r$ distinct letters \cite{gpt}, so we developed an algorithm to compute $\mathit{fw}(u)$ and we used $\mathit{fw}(u)$ to prove sharp bounds on $\mathit{Ex}(u, n)$ for several families of forbidden sequences $u$. We found that $\mathit{fw}((1 2 \ldots l)^{t})=2t-1$ for all $l\geq 2$ and $t\geq 1$, and thus $\mathit{Ex}((1 2 \ldots l)^{t}, n)=n2^{\frac{1}{(t-2)!}\alpha(n)^{t-2}\pm O(\alpha(n)^{t-3})}$ for all $l \geq 2$ and $t\geq 3$. This improved the upper bound from \cite{fox} on the number of edges in $k$-quasiplanar graphs with no pair of edges intersecting in more than $O(1)$ points.

In \cite{gtseq}, we developed a faster algorithm for computing formation width and used it to identify every sequence $u$ for which $u$ contains $ababa$ and $\mathit{fw}(u) = 4$. As a corollary, we showed that $\mathit{Ex}(u, n)=\Theta(n\alpha(n))$ for every such sequence $u$, which is the same as $\mathit{Ex}(a b a b a, n)$ up to a constant factor. Define a sequence $u$ to be \emph{minimally non-linear} if $\Ex(u, n) = \omega(n)$ but $\Ex(u', n) = O(n)$ for all sequences $u'$ properly contained in $u$. Pettie \cite{petmnl} posed the problem of determining every minimally non-linear sequence $u$. Every known minimally non-linear sequence $u$ has formation width $4$, so identifying sequences with formation width $4$ provides candidates for other minimally non-linear sequences. In \cite{GT1}, we used $\mathit{fw}(u)$ to find tight bounds on the lengths of generalized Davenport-Schinzel sequences that avoid $a b c (a c b)^t a b c$, answering a problem from \cite{gpt}, as well as generalized Davenport-Schinzel sequences that avoid $a b c a c b (a b c)^t a c b$. 

In this paper, we present an algorithm for computing $\fw(u)$ for sequences $u$ with at most $r$ distinct letters that has run time bounded by a polynomial in the length of $u$, where the degree of the polynomial depends on $r$. We also include an implementation of the algorithm in Python, and we compare its run time to the previous fastest algorithm for computing formation width. Furthermore, we apply the algorithm to obtain sharp bounds on $\Ex(u, n)$ for sequences $u$ with $3$ distinct letters that contain alternations of length $6$, $7$, $8$, or $9$. 

Given an alternation $u = (a b)^t$ of even length, let $u'$ be any sequence obtained by replacing each $b$ in $u$ with $bc$ or $cb$. So $(a b c)^t$ and $(a b c)^{t-1} a c b$ are two of the $2^t$ possibilities for $u'$. One might expect that $\Ex(u', n) = \omega(\Ex(u, n))$ in general, but we found in \cite{gpt} that $\fw((a b c)^t) = \fw((a b c)^{t-1} a c b) = \fw((a b)^t) = 2t-1$, which implies that $\Ex((a b c)^t, n) = n2^{\frac{1}{(t-2)!}\alpha(n)^{t-2}\pm O(\alpha(n)^{t-3})}$ and $\Ex((a b c)^{t-1} a c b, n) = n2^{\frac{1}{(t-2)!}\alpha(n)^{t-2}\pm O(\alpha(n)^{t-3})}$ for $t \geq 3$. A natural problem is to identify the sequences $u'$ for which $\Ex(u', n) = \Theta(\Ex(u, n))$, as well as the more open-ended problem of identifying the sequences $u'$ for which $\Ex(u', n) \approx \Ex(u, n)$. The latter problem makes more sense than the former for $t \geq 4$, since the bounds on $\Ex((ab)^t, n)$ for $t \geq 4$ are not tight up to a constant factor. We approach this problem by using the new formation width algorithm to determine the sequences $u'$ for which $\fw(u') = \fw(u)$ for $t \leq 12$, and proving the next theorem for all $t \geq 3$.

\begin{thm}\label{mainth}
If $u$ is one of the sequences $(a b c)^{t}$, $a b c (a c b)^{t-1}$, $a b c (a c b)^{t-2} a b c$, $a b c a c b (a b c)^{t-3} a c b$, $a b c a c b (a b c)^{t-2}$, $a b c a b c (a c b)^{t-2}$, $(a b c)^{t-2} a c b a c b$, $(a b c)^{t-2} a c b a b c$, or $(a b c)^{t-1} a c b$, then $\fw(u) = \fw((a b)^t)$ and $\Ex(u, n) = n2^{\frac{1}{(t-2)!}\alpha(n)^{t-2}\pm O(\alpha(n)^{t-3})}$ for all $t\geq 3$.
\end{thm}

\section{Algorithms for $\fw(u)$}

There are three algorithms for formation width including the new one, which we call \textit{BinaryFormation}, \textit{FormationTree}, and \textit{PermutationVector}.  As mentioned in the introduction, $\fw(u)$ is defined as the minimum value of $s$ for which there exists an $r$ such that every $(r, s)$-formation contains $u$. From the definition, it might seem like we need to check all integers $r > 0$ to know that there is no $(r, s)$-formation that contains $u$ for a given $s$. However, we showed in the original paper on $\fw(u)$ that only finitely many values of $r$ need to be checked for each $s$ \cite{gpt}. This gives the first algorithm for computing formation width.

\subsection{BinaryFormation}
To state the next result, we define an $(r, s)$-formation $f$ to be \emph{binary} if there exists a permutation $p$ on $r$ letters such that every permutation in $f$ is either equal to $p$ or the reverse of $p$.

\begin{thm}\cite{gpt}
For any sequence $u$ with $r$ distinct letters, $\fw(u)$ is the minimum value of $s$ for which every binary $(r, s)$-formation contains $u$. 
\end{thm}

This theorem gives an obvious algorithm for computing $\fw(u)$:

\begin{enumerate}
\item Let $r$ be the number of distinct letters in $u$ and let $s = 0$.
\item Check if every binary $(r, s)$-formation contains $u$. If so, return $s$. If not, increment $s$ and repeat this step.
\end{enumerate}

In the worst case, BinaryFormation has exponential run time in the length of $u$, so it is not practical for computing formation width of long sequences. In \cite{gpt} we showed that $\fw(u) = t$ for any sequence $u$ of length $t+1$ with two distinct letters. For sequences $u$ with three distinct letters, no polynomial time algorithm was known for computing $\fw(u)$. 

\subsection{FormationTree}

In \cite{gtseq}, we found a faster algorithm to compute $\fw(u)$, which allowed us to find all sequences $u$ for which $\fw(u) = 4$ and $u$ contains $a b a b a$. This answered a question that we asked in \cite{gpt}. The faster algorithm uses the binary tree structure of binary $(r, s)$-formations, where the root of the tree is the binary $(r, 1)$-formation with first permutation $1 \dots r$, and the binary $(r, s)$-formation $f$ is the parent of $f 1 \dots r$ and $f r \dots 1$. The main new idea of this algorithm was that we do not have to check the descendants of a binary $(r, s)$-formation $f$ once we know that $f$ contains $u$. 

\begin{enumerate}
\item Let $r$ be the number of distinct letters in $u$ and let $s = 1$. If $u$ has length $0$, return $0$. 
\item Let Fset be the set that contains only the $(r, 1)$-formation with first permutation $1 \dots r$.
\item Check if every binary $(r, s)$-formation in Fset contains $u$. If so, return $s$. If not, construct two new binary $(r, s+1)$-formations $f 1 \dots r$ and $f r \dots 1$ for each $f$ in Fset that avoids $u$, let Fset be the set of newly constructed formations, increment $s$, and repeat this step.
\end{enumerate}

Although FormationTree is much faster than BinaryFormation in practice, in the worst case it still has exponential run time in the length of $u$. For example, it is easy to see immediately that $\fw((a b c)^t) = 2t-1$ using the pigeonhole principle for the upper bound and an alternating binary formation for the lower bound, but this is very slow to compute with the FormationTree algorithm since we do not find any binary $(3, s)$-formations that contain $(a b c)^t$ until $s = t$, at which point we are checking for containment of $(a b c)^t$ in $2^{t-1}$ binary formations.

\subsection{PermutationVector}

Given a sequence $u$ with $r$ distinct letters $1, \dots, r$ and length $n$, PermutationVector maintains a dynamic set of vectors, where each vector has $r!$ entries corresponding to permutations of the distinct letters of $u$ and each entry has a value between $1$ and $n$. Each round, we modify the set of vectors, and the output of $\fw(u)$ is the number of rounds it takes for the set of vectors to become empty. 

In this algorithm, we no longer keep track of any formations as we increment $s$, since there are exponentially many binary formations in terms of $s$. Instead, for each binary formation $f$ on the same letters as $u$, we define a vector $p_f(u)$ with an entry for each permutation $\pi$ of the distinct letters of $u$. If $v$ is the sequence obtained from applying $\pi$ to the letters of $u$, then the value of the $\pi$ entry of $p_f(u)$ is the maximum length of an initial segment of $v$ that is a subsequence of $f$. There are at most $n^{r!}$ possible vectors in any round of the algorithm. We can tell whether $f$ contains $u$ by checking if any of the entries of $p_f(u)$ is equal to the length of $u$. 

Given a $p_f(u)$ with no entries equal to the length of $u$, we can compute $p_{f 1 \dots r}(u)$: if $v$ is the sequence obtained from applying the permutation $\pi$ to the letters of $u$, we use the $\pi$ entry of $p_f(u)$ to determine the longest initial segment $v^{*}$ of $v$ that is a subsequence of $f$ and we let $v'$ be obtained from $v$ by removing $v^{*}$, and then we find the longest initial segment of $v'$ that is a subsequence of $1 \dots r$ and we add its length to the length of $v^{*}$ to get the $\pi$ entry of $p_{f 1 \dots r}(u)$. We also compute $p_{f r \dots 1}(u)$ similarly: we find the longest initial segment of $v'$ that is a subsequence of $r \dots 1$, and we add its length to the length of $v^{*}$ to obtain the $\pi$ entry of $p_{f r \dots 1}(u)$.

\begin{enumerate}
\item Let $r$ be the number of distinct letters in $u$ and let $s = 1$. If $u$ has length $0$, return $0$. 
\item Let Vset be the set that contains only $p_{1 \dots r}(u)$.
\item Check if every vector in Vset has an entry equal to the length of $u$. If so, return $s$. If not, delete any vectors with entries equal to the length of $u$. For any vector $p_f(u)$ with no entries equal to the length of $u$, construct the two new vectors $p_{f 1 \dots r}(u)$ and $p_{f r \dots 1}(u)$ as described, let Vset be the set of newly constructed permutation vectors, increment $s$, and repeat this step.
\end{enumerate}

At each value of $s$ from $1$ to $\fw(u)$, we only check and replace at most $n^{r!}$ vectors, each with $r!$ entries. Computing each entry of $p_{f 1 \dots r}(u)$ and $p_{f r \dots 1}(u)$ from the corresponding entry in $p_f(u)$ takes $O(n)$ run time per entry, since finding the longest initial segment of $v'$ that is a subsequence of $1 \dots r$ or $r \dots 1$ takes $O(r)$ run time, and adding its length to the length of $v^{*}$ takes $O(\log{n})$ run time. Thus the run time is $O(n^{r!+2})$ when $r = O(1)$.

\subsection{Comparison of FormationTree and PermutationVector}

We compared the run times of FormationTree and PermutationVector on sequences $u$ with three and four distinct letters. We implemented PermutationVector in Python \cite{pvpy}, and we ran this against the Python implementation of FormationTree from \cite{gtseq}. The computations were performed on an ASUS TUF Gaming FX504 with a 2.30GHz Intel i5-8300H CPU and 8GB of RAM, running Python 3.7.1 on Windows 10. We ran the two algorithms on sequences of the form $(a b c)^t$ and $(a b c d)^t$ for $1 \leq t \leq 10$. For each sequence, we performed twenty trials of each of the algorithms and computed the mean run time. FormationTree was faster for four of the sequences: $a b c$, $a b c a b c$, $a b c d$, and $a b c d a b c d$. For all other sequences, PermutationVector was faster. For computing $\fw((a b c)^{10})$, PermutationVector was almost $300$ times faster than FormationTree.

\begin{table}\label{uniform1000}
\begin{center}
\caption{FormationTree (FT) versus PermuationVector (PV): comparison of run time in seconds for sequences of the form $(a b c)^t$ and $(a b c d)^t$}
\begin{tabular}{ |c|c|c|c|c| } 
 \hline
Sequence & Mean (FT) & Mean (PV) \\
\hline
$(a b c)$  &  0.000004  &  0.000014  \\
$(a b c)^2$  &  0.000056  &  0.000067  \\
$(a b c)^3$  &  0.00035  &  0.000229  \\
$(a b c)^4$  &  0.001525  &  0.00061  \\
$(a b c)^5$  &  0.006406  &  0.001413  \\
$(a b c)^6$  &  0.026046  &  0.002916  \\
$(a b c)^7$  &  0.104506  &  0.005198  \\
$(a b c)^8$  &  0.403432  &  0.008571  \\
$(a b c)^9$  &  1.5426  &  0.013164  \\
$(a b c)^{10}$  &  6.197161  &  0.021925  \\

$(a b c d)$  &  0.000005  &  0.000065  \\
$(a b c d)^2$  &  0.000211  &  0.000316  \\
$(a b c d)^3$  &  0.001483  &  0.001117  \\
$(a b c d)^4$  &  0.007846  &  0.003879  \\
$(a b c d)^5$  &  0.038453  &  0.012018  \\
$(a b c d)^6$  &  0.171947  &  0.032365  \\
$(a b c d)^7$  &  0.752102  &  0.080562  \\
$(a b c d)^8$  &  3.189359  &  0.186831  \\
$(a b c d)^9$  &  13.48486  &  0.476468  \\
$(a b c d)^{10}$  &  60.345885  &  0.778753  \\
 \hline

\end{tabular}
\end{center}
\end{table}

\section{Applications}

We computed all sequences with $3$ distinct letters that have formation width $x$ and alternation length $x+1$ for each $x = 5, 6, 7, 8$ using PermutationVector. These sequences are in the appendices \ref{x = 5}, \ref{x = 6}, \ref{x = 7}, and \ref{x = 8}, all on the alphabet $0, 1, 2$ with the letters making first appearances in that order. Combining $\Ex(u, n) = O(\F_{\fl(u),\fw(u)}(n))$ with the bounds on $\F_{r, s}(n)$, we obtain the following theorem.

\begin{thm}
\begin{enumerate}
\item For all of the sequences $u$ in Appendix \ref{x = 5}, $\Ex(u, n) = \Theta(n 2^{\alpha(n)})$. 

\item For all of the sequences $u$ in Appendix \ref{x = 6}, $\Ex(u, n) = \Theta(n \alpha(n) 2^{\alpha(n)})$. 

\item For all of the sequences $u$ in Appendix \ref{x = 7}, $\Ex(u, n) = n 2^{\frac{\alpha(n)^2}{2}\pm O(\alpha(n))}$.

\item For all of the sequences $u$ in Appendix \ref{x = 8}, $\Ex(u, n) = O(n2^{\frac{\alpha(n)^{2}}{2}(\log{\alpha(n)}+O(1))})$ and $\Ex(u, n) = \Omega(n2^{\frac{\alpha(n)^2}{2} - O(\alpha(n))})$. 
\end{enumerate} 
\end{thm}

In addition, we computed $\fw(u)$ for the more restricted family of sequences $u$ obtained by concatenating copies of $a b c$ and $a c b$, where the first permutation is $a b c$. We used PermutationVector to find all such sequences $u$ that have formation width $x$ and alternation length $x+1$ for $x = 5, 7, 9, \dots, 23$. They are listed in Appendix \ref{abcform}. 

For each $x = 2t+5 \geq 11$, there are nine sequences $u$ that have formation width $x$ and alternation length $x+1$ such that $u$ is obtained by concatenating copies of $a b c$ and $a c b$, where the first permutation is $a b c$. These sequences are $(a b c)^{t+3}$, $a b c (a c b)^{t+2}$, $a b c (a c b)^{t+1} a b c$, $a b c a c b (a b c)^{t} a c b$, $a b c a c b (a b c)^{t+1}$, $a b c a b c (a c b)^{t+1}$, $(a b c)^{t+1} a c b a c b$, $(a b c)^{t+1} a c b a b c$, and $(a b c)^{t+2} a c b$.

For $x = 5$ and $x = 7$, there are four and eight sequences respectively (all possible sequences obtained by concatenating copies of $a b c$ and $a c b$, where the first permutation is $a b c$). Interestingly, for $x = 9$ there are ten such sequences, even though there are nine for $x > 9$. Nine of the sequences for $x = 9$ have the same form as the nine sequences for each $x > 9$; the extra sequence for $x = 9$ is $a b c a b c a c b a c b a b c$.

In \cite{gpt}, we showed that $\fw((a b c)^t) = 2t-1$ and $\fw(a b c (a c b)^{t}) = 2t+1$ for $t \geq 0$. Furthermore in \cite{GT1} we proved that $\fw(a b c (a c b)^{t} a b c) = 2t+3$ and $\fw(a b c a c b (a b c)^t a c b) = 2t+5$ for $t \geq 0$. We use a computational method to handle the five other cases that appear for each $x > 9$ in Appendix \ref{abcform}, automating much of the proof with a Python script \cite{acbpy}. Also we note that the next result implies Theorem \ref{mainth} as a corollary.

\begin{prop}\label{halfmhalfc}
If $t \geq 1$ and $u$ is one of the sequences $a b c a c b (a b c)^{t}$, $a b c a b c (a c b)^{t}$, $(a b c)^{t} a c b a c b$, $(a b c)^{t} a c b a b c$, or $(a b c)^{t+1} a c b$, then $\fw(u) = 2t+3$.
\end{prop}

\begin{proof}
The lower bound $\fw(u) \geq 2t+3$ is immediate, so we prove that $\fw(u) \leq 2t+3$. By symmetry, it suffices to prove for $t \geq 1$ that $\fw(u) \leq 2t+3$ for each sequence $u = (a b c)^{t}b a c a b c$, $(a b c)^{t} b a c b a c$, $(a b c)^{t} a c b a c b$, $(a b c)^{t} a c b a b c$, or $(a b c)^{t+1} a c b$. Let $f$ be any binary $(3, 2t+3)$-formation with permutations $x y z$ and $z y x$. Then the first $2t-1$ permutations of $f$ have the subsequence $(x y z)^t$ or $(z y x)^t$. Without loss of generality, suppose the first $2t-1$ permutations of $f$ have the subsequence $(x y z)^t$. 

We consider two cases. For the first case, suppose that the first $2t-1$ permutations of $f$ have the subsequence $(x y z)^{t+1}$. This immmediately implies that $f$ has the subsequences $(x y z)^t x z y x z y$, $(x y z)^t x z y x y z$, and $(x y z)^{t+1} x z y$. Moreover, $f$ has the subsequence $(x y z)^t y x z x y z$ unless its last $12$ letters are $z y x x y z z y x z y x$, in which case it has the subsequence $(y z x)^t z y x y z x$. Furthermore, $f$ has the subsequence $(x y z)^t y x z y x z$ unless its last $12$ letters are $z y x x y z x y z z y x$, in which case it has the subsequence $(y z x)^t z y x z y x$. Thus we have shown that any binary $(3, 2t+3)$-formation that has the subsequence $(x y z)^{t+1}$ in its first $2t-1$ permutations contains $u$ for each  $u = (a b c)^{t}b a c a b c$, $(a b c)^{t} b a c b a c$, $(a b c)^{t} a c b a c b$, $(a b c)^{t} a c b a b c$, or $(a b c)^{t+1} a c b$.

For the second case, suppose that the first $2t-1$ permutations of $f$ have the subsequence $(x y z)^t$ and the subsequence $(z y x)^{t-1}$. We handle each sequence with a simple computation \cite{acbpy} in Python described below. 

\begin{enumerate}

\item For $u = (a b c)^{t}b a c a b c$, we use \cite{acbpy} to verify that every binary $(3, 4)$-formation on the permutations $xyz$ and $zyx$ except for $xyzxyzzyxzyx$ has at least one of the following subsequences: $y x z x y z$, $x z y x y z x$, $x y x z y z x y$, $z y x y z x z y x$, $z y x z x y z y x z$, $z y x z y z x y x z y$. So we may assume that the last $12$ letters of $f$ are $xyzxyzzyxzyx$. If permutation $2t-1$ is $zyx$, then $f$ has the subsequence $(x y z)^t y x z x y z$. If permutation $2t-1$ is $x y z$, then $f$ has the subsequence $(z y x)^t y z x z y x$.

\item For $u = (a b c)^t b a c b a c$, we use \cite{acbpy} to verify that every binary $(3, 4)$-formation on the permutations $xyz$ and $zyx$ except for $xyz zyx xyz xyz$ has at least one of the following subsequences: $y x z y x z$, $x z y x z y x$, $x y x z y x z y$, $z y x y z x y z x$, $z y x z x y z x y z$, $z y x z y z x y z x y$. So we may assume that the last $12$ letters of $f$ are $xyz zyx xyz xyz$. If permutation $1$ is $zyx$, then $f$ has the subsequence $(z x y)^t x z y x z y$. If permutation $2t-1$ is $z y x$, then $f$ has the subsequence $(x y z)^t y x z y x z$. If permutations $1$ and $2t-1$ are both $x y z$, then $f$ has the subsequence $(x z y)^t z x y z x y$.

\item For $u = (a b c)^t a c b a c b$, we use \cite{acbpy} to verify that every binary $(3, 4)$-formation on the permutations $xyz$ and $zyx$ except for $zyx xyz zyx xyz$ and $zyx xyz xyz zyx$ has at least one of the following subsequences: $x z y x z y$, $x y x z y x z$, $x y z y x z y x$, $z y x z x y z x y$, $z y x z y z x y z x$, $z y x z y x y z x y z$. So we may assume that the last $12$ letters of $f$ are $zyx xyz zyx xyz$ or $zyx xyz xyz zyx$. If permutation $1$ is $zyx$, then in both cases $f$ has the subsequence $(z x y)^t z y x z y x$. If permutation $1$ is $xyz$, then in both cases $f$ has the subsequence $(x z y)^t x y z x y z$. 

\item For $u = (a b c)^t a c b a b c$, we use \cite{acbpy} to verify that every binary $(3, 4)$-formation on the permutations $xyz$ and $zyx$ except for $zyx xyz zyx zyx$ has at least one of the following subsequences: $x z y x y z$, $x y x z y z x$, $x y z y x z x y$, $z y x z x y z y x$, $z y x z y z x y x z$, $z y x z y x y z x z y$. So we may assume that the last $12$ letters of $f$ are $zyx xyz zyx zyx$. If permutation $1$ is $zyx$, then $f$ has the subsequence $(z x y)^t z y x z x y$. If permutation $1$ is $x y z$, then $f$ has the subsequence $(x z y)^t x y z x z y$.

\item For $u = (a b c)^{t+1} a c b$, we use \cite{acbpy} to verify that every binary $(3, 4)$-formation on the permutations $xyz$ and $zyx$ except for $zyx zyx xyz zyx$ and $zyx xyz zyx xyz$ has at least one of the following subsequences: $x y z x z y$, $x y z x y x z$, $x y z x y z y x$, $z y x z y x z x y$, $z y x z y x z y z x$, $z y x z y x z y x y z$. So we may assume that the last $12$ letters of $f$ are $zyx zyx xyz zyx$ or $zyx xyz zyx xyz$. If permutation $1$ is $zyx$, then in both cases $f$ has the subsequence $(z x y)^{t+1} z y x$. If permutation $1$ is $x y z$, then in both cases $f$ has the subsequence $(x z y)^{t+1} x y z$.

\end{enumerate}
\end{proof}

We finish with a conjecture based on Proposition \ref{halfmhalfc} and the evidence in Appendix \ref{abcform}.

\begin{conj}\label{acbconj}
Suppose that $t \geq 6$. Among $(3, t)$-formations $u$ with first permutation $a b c$ and other permutations equal to $a b c$ or $a c b$, $\fw(u) = 2t-1$ if and only if $u$ is $(a b c)^{t}$, $a b c (a c b)^{t-1}$, $a b c (a c b)^{t-2} a b c$, $a b c a c b (a b c)^{t-3} a c b$, $a b c a c b (a b c)^{t-2}$, $a b c a b c (a c b)^{t-2}$, $(a b c)^{t-2} a c b a c b$, $(a b c)^{t-2} a c b a b c$, or $(a b c)^{t-1} a c b$.
\end{conj}

For any fixed $r$, PermutationVector has worst-case run time polynomial in $n$ for all sequences of length $n$ with at most $r$ letters, where the degree of the polynomial depends on $r$. However, PermutationVector can be slower than FormationTree when $n$ is not sufficiently large relative to $r$, as we saw with the sequences $abc$, $abcd$, $abcabc$, and $abcdabcd$. Clearly PermutationVector does not have run time polynomial in $n$ when $r = \Theta(n)$, since the vectors in the algorithm each have $r!$ entries.

\begin{ques}
Is there an algorithm for computing $\fw(u)$ that has worst-case run time polynomial in $n$ for all sequences $u$ of length $n$, regardless of the number of distinct letters in $u$?
\end{ques}

\appendix

\section{$3$-letter sequences with formation width $5$ and alternation length $6$}\label{x = 5}

0101012,
0101021,
0101201,
0102101,
0120101,
0120202,
0121212,
01012012,
01012021,
01012201,
01021012,
01021021,
01021201,
01022101,
01201012,
01201021,
01201201,
01201202,
01201212,
01202012,
01202021,
01202101,
01202102,
01210202,
01210212,
01212012,
01212102,
01220101,
012012012,
012012021,
012012102,
012021012,
012021021,
012021102,
012102012,
012102102,
012201021

\section{$3$-letter sequences with formation width $6$ and alternation length $7$}\label{x = 6}

01010102,
01010120,
01010210,
01012010,
01021010,
01201010,
01202020,
01212121,
010102102,
010102120,
010102210,
010120102,
010120120,
010120210,
010122010,
010210102,
010210120,
010210210,
010212010,
010221010,
012010102,
012010120,
012010210,
012012010,
012012020,
012012121,
012020120,
012020210,
012021010,
012021020,
012102020,
012102121,
012120121,
012121021,
012121210,
012201010,
0102102102,
0102102120,
0102120102,
0102120120,
0102210120,
0120102102,
0120102120,
0120102210,
0120120102,
0120120120,
0120120121,
0120120210,
0120121020,
0120121021,
0120121210,
0120210120,
0120210210,
0120211020,
0121020120,
0121020121,
0121021020,
0121021021,
0121210021,
0121210210,
0122010102,
0122010210

\section{$3$-letter sequences with formation width $7$ and alternation length $8$} \label{x = 7}

010101012,
010101021,
010101201,
010102101,
010120101,
010210101,
012010101,
012020202,
012121212,
0101012012,
0101012021,
0101012201,
0101021012,
0101021021,
0101021201,
0101022101,
0101201012,
0101201021,
0101201201,
0101202101,
0101220101,
0102101012,
0102101021,
0102101201,
0102102101,
0102120101,
0102210101,
0120101012,
0120101021,
0120101201,
0120102101,
0120120101,
0120120202,
0120121212,
0120201202,
0120202012,
0120202021,
0120202102,
0120210101,
0120210202,
0121020202,
0121021212,
0121201212,
0121210212,
0121212012,
0121212102,
0122010101,
01012012012,
01012012021,
01012021012,
01012021021,
01012201021,
01021012012,
01021012021,
01021012201,
01021021012,
01021021021,
01021021201,
01021201021,
01021201201,
01022101012,
01022101201,
01201012012,
01201012021,
01201012201,
01201021012,
01201021021,
01201021201,
01201022101,
01201201012,
01201201021,
01201201201,
01201201202,
01201201212,
01201202012,
01201202021,
01201202101,
01201202102,
01201210202,
01201210212,
01201212012,
01201212102,
01202012012,
01202012021,
01202012102,
01202021012,
01202021021,
01202021102,
01202101012,
01202101021,
01202101201,
01202101202,
01202102012,
01202102021,
01202102101,
01202102102,
01202110202,
01210201202,
01210201212,
01210202012,
01210202102,
01210210202,
01210210212,
01210212012,
01210212102,
01212012012,
01212012102,
01212100212,
01212102012,
01212102102,
01220101021,
01220101201,
01220102101,
012012012012,
012012012021,
012012012102,
012012021012,
012012021021,
012012021102,
012012102012,
012012102102,
012021012012,
012021012021,
012021021012,
012021021021,
012102012012,
012102012102,
012102102012,
012102102102,
012201021021

\section{$3$-letter sequences with formation width $8$ and alternation length $9$} \label{x = 8}

010101010,
0101010102,
0101010120,
0101010210,
0101012010,
0101021010,
0101201010,
0102101010,
0120101010,
0120202020,
0121212121,
01010102102,
01010102120,
01010102210,
01010120102,
01010120120,
01010120210,
01010122010,
01010210102,
01010210120,
01010210210,
01010212010,
01010221010,
01012010102,
01012010120,
01012010210,
01012012010,
01012021010,
01012201010,
01021010102,
01021010120,
01021010210,
01021012010,
01021021010,
01021201010,
01022101010,
01201010102,
01201010120,
01201010210,
01201012010,
01201021010,
01201201010,
01201202020,
01201212121,
01202012020,
01202020120,
01202020210,
01202021020,
01202101010,
01202102020,
01210202020,
01210212121,
01212012121,
01212102121,
01212120121,
01212121021,
01212121210,
01220101010,
010102102102,
010102102120,
010102120102,
010102120120,
010102210120,
010120102102,
010120102120,
010120102210,
010120120102,
010120120120,
010120120210,
010120210120,
010120210210,
010122010102,
010122010210,
010210102102,
010210102120,
010210102210,
010210120102,
010210120120,
010210120210,
010210122010,
010210210102,
010210210120,
010210210210,
010210212010,
010212010102,
010212010120,
010212010210,
010212012010,
010221010120,
010221010210,
010221012010,
012010102102,
012010102120,
012010102210,
012010120102,
012010120120,
012010120210,
012010122010,
012010210102,
012010210120,
012010210210,
012010212010,
012010221010,
012012010102,
012012010120,
012012010210,
012012012010,
012012012020,
012012012121,
012012020120,
012012020210,
012012021010,
012012021020,
012012102020,
012012102121,
012012120121,
012012121021,
012012121210,
012020120120,
012020120210,
012020121020,
012020210120,
012020210210,
012020211020,
012021010120,
012021010210,
012021012010,
012021012020,
012021020120,
012021020210,
012021021010,
012021021020,
012021102020,
012102012020,
012102012121,
012102020120,
012102021020,
012102102020,
012102102121,
012102120121,
012102121021,
012120120121,
012120121021,
012120121210,
012121002121,
012121020121,
012121021021,
012121021210,
012121210021,
012121210210,
012201010102,
012201010210,
012201012010,
012201021010,
0102102102102,
0102102102120,
0102102120102,
0102102120120,
0102120102102,
0102120102120,
0102120120102,
0102120120120,
0102210120120,
0120102102102,
0120102102120,
0120102120102,
0120102120120,
0120102210120,
0120120102102,
0120120102120,
0120120102210,
0120120120102,
0120120120120,
0120120120121,
0120120120210,
0120120121020,
0120120121021,
0120120121210,
0120120210120,
0120120210210,
0120120211020,
0120121020120,
0120121020121,
0120121021020,
0120121021021,
0120121210210,
0120210120120,
0120210120210,
0120210210120,
0120210210210,
0120211020120,
0121020120120,
0121020120121,
0121020121020,
0121020121021,
0121021020120,
0121021020121,
0121021021020,
0121021021021,
0121210210210,
0122010210210

\section{Evidence for Conjecture \ref{acbconj}}\label{abcform}
In this section, we list all sequences obtained by concatenating copies of $a b c$ and $a c b$, where the first permutation is $a b c$, that have formation width $x$ and alternation length $x+1$ for $x = 5, 7, 9, \dots, 23$. On each line, we list a binary string $s$, a sequence $u$, and $\fw(u)$. The sequence $u$ is obtained from $s$ by replacing each $0$ with $021$ and each $1$ with $012$. \\

\noindent 100 012021021 5  \\
101 012021012 5  \\
110 012012021 5  \\
111 012012012 5  \\

\noindent 1000 012021021021 7  \\
1001 012021021012 7  \\
1010 012021012021 7  \\
1011 012021012012 7  \\
1100 012012021021 7  \\
1101 012012021012 7  \\
1110 012012012021 7  \\
1111 012012012012 7 \\

\noindent 10000 012021021021021 9  \\
10001 012021021021012 9  \\
10110 012021012012021 9  \\
10111 012021012012012 9  \\
11000 012012021021021 9  \\
11001 012012021021012 9  \\
11100 012012012021021 9  \\
11101 012012012021012 9  \\
11110 012012012012021 9  \\
11111 012012012012012 9  \\

\noindent 100000 012021021021021021 11  \\
100001 012021021021021012 11  \\
101110 012021012012012021 11  \\
101111 012021012012012012 11  \\
110000 012012021021021021 11  \\
111100 012012012012021021 11  \\
111101 012012012012021012 11  \\
111110 012012012012012021 11  \\
111111 012012012012012012 11  \\

\noindent 1000000 012021021021021021021 13  \\
1000001 012021021021021021012 13  \\
1011110 012021012012012012021 13  \\
1011111 012021012012012012012 13  \\
1100000 012012021021021021021 13  \\
1111100 012012012012012021021 13  \\
1111101 012012012012012021012 13  \\
1111110 012012012012012012021 13  \\
1111111 012012012012012012012 13  \\

\noindent 10000000 012021021021021021021021 15  \\
10000001 012021021021021021021012 15  \\
10111110 012021012012012012012021 15  \\
10111111 012021012012012012012012 15  \\
11000000 012012021021021021021021 15  \\
11111100 012012012012012012021021 15  \\
11111101 012012012012012012021012 15  \\
11111110 012012012012012012012021 15  \\
11111111 012012012012012012012012 15  \\

\noindent 100000000 012021021021021021021021021 17  \\
100000001 012021021021021021021021012 17  \\
101111110 012021012012012012012012021 17  \\
101111111 012021012012012012012012012 17  \\
110000000 012012021021021021021021021 17  \\
111111100 012012012012012012012021021 17  \\
111111101 012012012012012012012021012 17  \\
111111110 012012012012012012012012021 17  \\
111111111 012012012012012012012012012 17  \\

\noindent 1000000000 012021021021021021021021021021 19  \\
1000000001 012021021021021021021021021012 19  \\
1011111110 012021012012012012012012012021 19  \\
1011111111 012021012012012012012012012012 19  \\
1100000000 012012021021021021021021021021 19  \\
1111111100 012012012012012012012012021021 19  \\
1111111101 012012012012012012012012021012 19  \\
1111111110 012012012012012012012012012021 19  \\
1111111111 012012012012012012012012012012 19  \\

\noindent 10000000000 012021021021021021021021021021021 21  \\
10000000001 012021021021021021021021021021012 21  \\
10111111110 012021012012012012012012012012021 21  \\
10111111111 012021012012012012012012012012012 21  \\
11000000000 012012021021021021021021021021021 21  \\
11111111100 012012012012012012012012012021021 21  \\
11111111101 012012012012012012012012012021012 21  \\
11111111110 012012012012012012012012012012021 21  \\
11111111111 012012012012012012012012012012012 21  \\

\noindent 100000000000 012021021021021021021021021021021021 23  \\
100000000001 012021021021021021021021021021021012 23  \\
101111111110 012021012012012012012012012012012021 23  \\
101111111111 012021012012012012012012012012012012 23  \\
110000000000 012012021021021021021021021021021021 23  \\
111111111100 012012012012012012012012012012021021 23  \\
111111111101 012012012012012012012012012012021012 23  \\
111111111110 012012012012012012012012012012012021 23  \\
111111111111 012012012012012012012012012012012012 23  \\
\end{document}